\documentclass[conference]{IEEEtran}
\IEEEoverridecommandlockouts

\usepackage{algorithm}
\usepackage{algorithmicx}
\usepackage{algpseudocode}
\usepackage{amsmath}
\usepackage{indentfirst}
\floatname{algorithm}{Algorithm}

\usepackage{hyperref}  
\usepackage{cite}
\usepackage{amsmath,amssymb,amsfonts}

\usepackage{graphicx}
\usepackage{textcomp}
\usepackage{xcolor}
\usepackage{epstopdf}
\usepackage{subfigure}
\usepackage{booktabs}
\usepackage{amsthm}
\usepackage{enumerate}
\usepackage{subfigure}

\usepackage{xcolor}
\algnewcommand\BlueKeyword{\textcolor{blue}}

\newtheorem{definition}{Definition}

\newtheorem{lemma}{Lemma}
\newtheorem{theorem}{Theorem}

\def\BibTeX{{\rm B\kern-.05em{\sc i\kern-.025em b}\kern-.08em
    T\kern-.1667em\lower.7ex\hbox{E}\kern-.125emX}}
\begin{document}
\title{Bi-LSTM based Multi-Agent DRL with Computation-aware Pruning for Agent Twins Migration in Vehicular Embodied AI Networks}
\author{Yuxiang Wei, Zhuoqi Zeng, Yue Zhong, Jiawen Kang*, Ryan Wen Liu, M. Shamim Hossain
\thanks{
Yuxiang Wei, Zhuoqi Zeng, Yue Zhong, Jiawen Kang are with the School of Automation, Guangdong University of Technology, Guangzhou 510006, China (e-mail: 3122001501@mail2.gdut.edu.cn; 3123001489@mail2.gdut.edu.cn; 2112404106@mail2.gdut.edu.cn; kavinkang@gdut.edu.cn).

Ryan Wen Liu is with the School of Navigation, Wuhan University of Technology, Wuhan 430063, China, and also with the State Key Laboratory of Maritime Technology and Safety, Wuhan 430063, China (e-mail: wenliu@whut.edu.cn).

M. Shamim Hossain is with the Department of Software Engineering, College of Computer and Information Sciences, King Saud University, Riyadh 12372, Saudi Arabia (e-mail: mshossain@ksu.edu.sa).

(\textit{*Corresponding author: Jiawen Kang}).
}
}

\maketitle

\begin{abstract}
With the advancement of large language models and embodied Artificial Intelligence (AI) in the intelligent transportation scenarios, the combination of them in intelligent transportation spawns the Vehicular Embodied AI Network (VEANs). In VEANs, Autonomous Vehicles (AVs) are typical agents whose local advanced AI applications are defined as vehicular embodied AI agents, enabling capabilities such as environment perception and multi-agent collaboration. Due to computation latency and resource constraints, the local AI applications and services running on vehicular embodied AI agents need to be migrated, and subsequently referred to as vehicular embodied AI agent twins, which drive the advancement of vehicular embodied AI networks to offload intensive tasks to Roadside Units (RSUs), mitigating latency problems while maintaining service quality. Recognizing workload imbalance among RSUs in traditional approaches, we model AV-RSU interactions as a Stackelberg game to optimize bandwidth resource allocation for efficient migration. A Tiny Multi-Agent Bidirectional LSTM Proximal Policy Optimization (TMABLPPO) algorithm is designed to approximate the Stackelberg equilibrium through decentralized coordination. Furthermore, a personalized neural network pruning algorithm based on Path eXclusion (PX) dynamically adapts to heterogeneous AV computation capabilities by identifying task-critical parameters in trained models, reducing model complexity with less performance degradation. Experimental validation confirms the algorithm’s effectiveness in balancing system load and minimizing delays, demonstrating significant improvements in vehicular embodied AI agent deployment.

\end{abstract}


\section{Introduction}
Embodied Artificial Intelligence (AI) is an innovative paradigm in AI that integrates perception, reasoning, and interactions within physical environments. In recent years, large models have driven notable progress in embodied AI \cite{Hong_2024_CVPR}, \cite{lin2024embodied}. Due to their superior sim-to-real adaptability, the large models have become “brains” of many embodied AI networks and systems \cite{NEURIPS2023_ee6630dc}. They also use few-shot learning to fit the embodied AI networks with real transportation scenes \cite{song2023llm}. The few-shot learning is a framework that trains the models to satisfy the users’ habits and requirements with small-scale samples. Specially, for intelligent transportation scenarios, Autonomous Vehicles (AVs), acting as vehicular embodied agents, interact with various on-board devices and roadside infrastructures to build Vehicular Embodied AI Networks (VEANs) \cite{sharma2024artificial, 10700687}. However, in VEANs, resource constraints of the AVs limit their ability to run latency-sensitive embodied AI applications, such as self-driving navigation and motion planning, which are defined as vehicular embodied AI agents \cite{10798474}. These applications should be offloaded to nearby resource-sufficient Roadside Units (RSUs) and managed through digital twin technologies. After offloading, the RSUs segment resources to generate the Vehicular Embodied Agent AI Twins (VEAATs) that establish AV-application mapping while simultaneously executing vehicular embodied AI tasks \cite{zhong2025generative}. Thus, the AVs can receive real-time feedback of their VEAATs from RSUs.

To ensure continuous high-quality services for AVs, the VEAATs on the current RSUs need to be migrated to the next RSUs decided by AVs because of the limited coverage of the current RSUs and dynamic mobility of the AVs \cite{zhong2023blockchain}. Considering the hotspots traffic conditions, the high-density resource demands from vehicles cause workload imbalances among RSUs, risking overloaded issues \cite{10533222}. To address these challenges, existing works exploit game theory \cite{kang2024metaverses}, contact theory \cite{zhong2025generative}, and learning based methods \cite{zhang2023learning} to optimize resource allocation and balance workload issues. However, these works ignore the dynamic interactions among multiple resource providers and requesters. And the Quality of Services (QoS) has not been considered adequately.

To address these challenges, we design a hierarchical bandwidth allocation strategy that is critical to balance AVs and RSUs, where RSUs initially optimize resource allocation while AVs adapt RSUs' strategies to select the optimal RSUs for VEAAT migrations and prioritize user experience metrics. This coordination enables dynamic adjustments through continuous agent-environment feedback loops. Considering the dynamic interaction between AVs and RSUs, we formulate a Multi-Leader Multi-Follower (MLMF) Stackelberg game between them. In this game, the RSUs act as leaders to reduce bandwidth pressure, while the AVs act as followers that aim to request enough bandwidth to minimize latency.

To enhance the model performance in complex problems, researchers have increasingly adopted advanced Deep Reinforcement Learning (DRL) algorithms, such as a multi-attribute double dutch auction-based mechanism with a DRL-based auctioneer \cite{tong2024deep}, an incentive mechanism with the multi-dimensional contract theory \cite{zhong2025generative}, and a consortium blockchain and federated multi-agent DRL integrated framework \cite{abishu2024blockchain}. However, these existing works often neglect the high computation cost, resulting in excessive usage of AVs’ computation resources and significant latency. Considering the complexity of the Stackelberg game and the temporal continuity of the interactions between AVs and RSUs, we employ Multi-Agent Bidirectional Long Short-Term Memory Proximal Policy Optimization (MABLPPO), a Multi-Agent Deep Reinforcement Learning (MADRL) algorithm that incorporates Bidirectional Long Short-Term Memory (Bi-LSTM) networks, to identify the equilibrium in the Stackelberg game \cite{huang2015bidirectional}. The MABLPPO can use the environment simulator to train the model instead of using a very small number of labeled examples. It can enhance the performance of models and generalization ability in resource optimization tasks. This algorithm divides AVs and RSUs into two parts. Both can optimize their decisions based on their local observations, including the environmental changes and actions. Additionally, to ensure that AVs can execute their actions within a limited time and reduce the computation resources allocated from AVs, we employ the Path eXclusion (PX) pruning algorithm to accelerate the processing \cite{iurada2024finding}. With computation-aware pruning, each AV can maintain real-time performance to handle different problems in a dynamic transport environment.

Considering the above analysis, we propose a Tiny MABLPPO (TMABLPPO) algorithm to approximate the Stackelberg Equilibrium (SE). The main contributions of this paper are summarized as follows:

\begin{itemize}
\item Considering the competitive interactions between AVs and RSUs, we formulate a multi-leader multi-follower Stackelberg game between them to optimize bandwidth resources. Specially, we integrate QoS metrics into utility functions, enabling quantitative modeling of service-level satisfaction and dynamic equilibrium analysis in vehicular embodied AI networks.
\item To effectively capture bidirectional temporal dependencies in interactions and improve the strategic performance of AVs and RSUs in the Stackelberg game, we propose a network named MABLPPO algorithm. Bi-LSTM algorithm perceives prev-and-post environmental state changes to enhance robustness in the VEAAT migration.
\item We design a computation-aware pruning algorithm based on PX pruning algorithm to address the demand for low latency and the difference in AVs’ computation resources, accelerating AVs and RSUs to make optimal decisions through MABLPPO. In contrast to the existing approaches, this algorithm innovates to find a balance between accuracy and latency.
\end{itemize}

The rest of this paper is structured as follows. Section \ref{related_works} discusses the related works. Section \ref{system_model} presents the system model. Section \ref{equilibrium_analysis} shows the equilibrium analysis for the MLMF Stackelberg game. Section \ref{MABLPPO_with_pruning} details the Bi-LSTM-based MADRL algorithm with pruning. Section \ref{numerical_results} demonstrates numerical results, and the conclusion is presented in Section \ref{conclusion}.

\begin{figure*}
    \centering
    \includegraphics[width=0.75\linewidth]{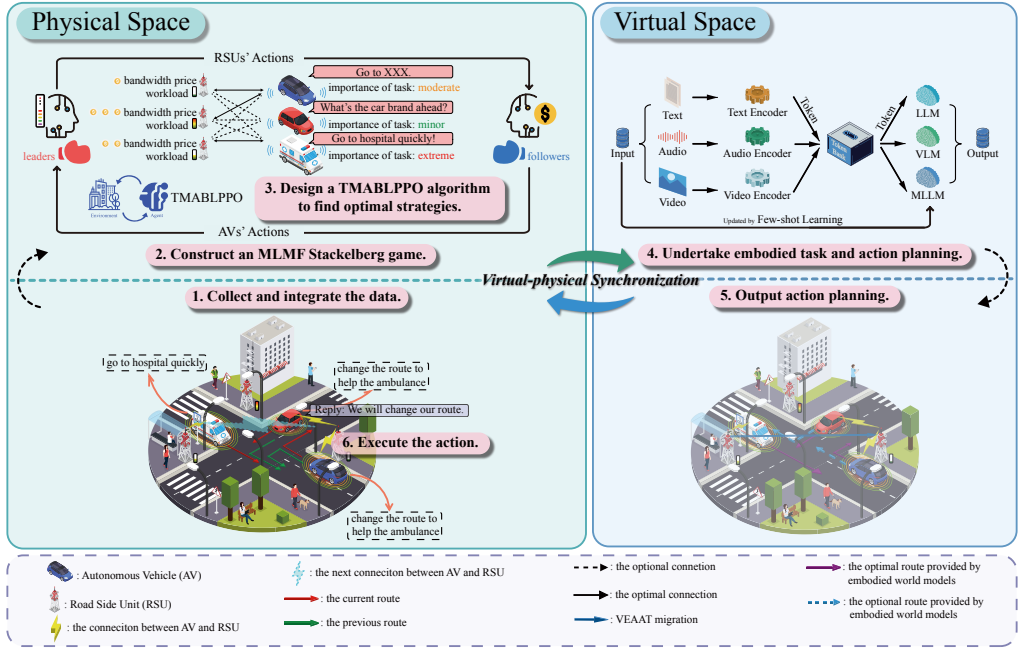}
    \caption{The system model for VEAAT migration. }
    \label{systemfig}
\end{figure*}

\section{Related Works} \label{related_works}
\subsection{Embodied AI in Vehicular Networks}

The rapid progress in computing capabilities and large-model technologies has driven interest in integrating embodied AI into vehicles, particularly for autonomous driving, as noted by Wang et al. \cite{wang2022conversational}. Embodied AI enables vehicles to perceive and interact with the physical world. For instance, embodied AI enhances human-computer interaction in vehicles, enabling voice-based control of vehicle functions. Liu et al. \cite{liu2024aligning} highlighted the potential of embodied AI to improve artificial general intelligence through cognitive capabilities. Moreover, Zhou et al. \cite{zhou2024embodied} further demonstrated its applications in activity prediction and situational analysis. However, in the context of the VEAN, the high computation demands of vehicular embodied AI may not be handled on the AVs. Hence, it is necessary to offload vehicular embodied AI tasks to RSUs for high-performance edge computing, but existing studies have yet to address migration challenges for seamless deployment.

Due to the different traffic conditions in temporal contexts, the embodied AI in vehicular networks should dynamically extract the features of intelligent transportation. Sural et al. \cite{sural2024contextvlm} found that the ContextVLM had a better performance with few-shot learning, which showed that the few-shot approach with VLMs is useful for generalization with a small number of annotated examples. Song et al. \cite{song2023comprehensive} use few-shot learning to enable AV to identify and respond to new traffic scenarios with minimal or even no prior data. Similarly, Chen et al. \cite{chen2024fine} highlighted the growing adoption of few-shot learning due to their proven efficacy in fine-tuning pre-trained models for specialized downstream tasks in intelligent transportation systems. To enhance the robustness of embodied AI in intelligent transportation systems, we incorporated few-shot learning capabilities to enable dynamic adaptation to evolving different intelligent transportation conditions.

\subsection{Service Migration with Stackelberg Game}

Service migration presents challenges in multi-party systems that require efficient resource scheduling and allocation, which is also an NP-hard problem \cite{sun2016primal}. Game theories, particularly the Stackelberg game, have been employed to construct a mathematical model to approximate the optimal solution in service migration. The Stackelberg game can effectively optimize the strategies of leaders and followers, making it well-suited for service migration scenarios. For example, Kang et al. \cite{kang2024metaverses} used the Stackelberg game to improve data transfer efficiency between Metaverse Service Providers and Metaverse Resource Providers to enhance the vehicular twin migration. Chen et al. \cite{10185562} proposed a framework for migration services, categorizing scenarios into urban areas (high-density networks) and remote areas (sparse connectivity) to address heterogeneous environmental constraints. Similarly, Zhang et al. \cite{zhang2025stackelberg} utilized the Stackelberg framework to optimize resource allocation in the collaborative intelligent transportation systems. The above works highlight the resilience of Stackelberg game theory for hierarchical decision-making, especially for dealing with constrained resources and multiple players. With the continued progress in computing, networking, and the Internet of Things technologies, the Stackelberg game is anticipated to play an increasingly vital role in service migration. With the development of the VEAN, there will be an increasing number of AI applications spanning diverse domains, each exhibiting varying levels of importance and distinct latency requirements. However, recent studies overlooked the varying importance of tasks or the distinct latency requirements associated with each task while migrating tasks.

\subsection{Model Compression for DRL}
Due to the discounted reward mechanisms of DRL, researchers usually use DRL to solve resource optimization models. To address computation challenges in deploying DRL models on resource-constrained devices, pruning has emerged as a pivotal technique for eliminating redundant neurons and weights while retaining performance. Kang et al. \cite{kang2024tiny} introduced structured pruning by evaluating neuron importance to reduce the actor network size. Livne et al. \cite{livne2020pops} proposed Policy Pruning and Shrinking (PoPS), leveraging weight rankings to effectively remove low-impact connections and compress models. Further advancing adaptive sparsity, Camci et al. \cite{camci2022qlp} integrated deep Q-learning to determine layer-wise sparsity ratios, enabling unstructured magnitude-based pruning dynamically. These approaches reduce computation complexity and storage demands, making DRL feasible for embedded systems and edge devices. Complementary methods like knowledge distillation. Wang et al. \cite{wang2021knowledge} proposed a Knowledge Distillation-based Cooperative Reinforcement Learning framework for offering connectivity flexibility in dynamic unmanned aerial vehicle networks.

Model compression algorithms like pruning and knowledge distillation compress models while retaining most of their original performance. These are particularly advantageous for DRL tasks that operate in traffic settings with low-latency requirements. The combination of pruning, which eliminates unnecessary neurons, and knowledge distillation, which is the distillation of essential knowledge into small models, enables DRL systems to be effectively implemented in low-computation-resource settings without hugely affecting efficiency and robustness in decision-making. To address practical deployment constraints, unlike prior fixed compression strategies, we specifically find pruning ratios to match the heterogeneous computation capabilities of individual AVs, ensuring adaptability across diverse hardware environments.

\section{System Model} \label{system_model}

\subsection{Migration Model in Embodied AI}
Due to limited onboard computing resources, the AVs cannot tackle all tasks locally. To alleviate the workload of AVs, the AVs need to offload tasks to nearby RSUs \cite{10185562}. Besides, integrating multi-source sensory data from surrounding AVs enhances the precision of vehicle control and enables cooperative vehicle-road coordination. Thus, the AVs need to find the optimal RSUs to connect and interact with the embodied world models, enabling real-time synchronization between physical and virtual spaces. Furthermore, the dynamic nature of the traffic environment demands adaptive and efficient resource allocation strategies for enabling the real-time execution of tasks. The optimization of RSU selection enables AVs to minimize communication latency, thereby improving overall system performance \cite{zhong2025generative}. In this paper, we utilize the Stackelberg game to address resource optimization between AVs and RSUs scenarios. Figure \ref{systemfig} illustrates the migration workflows while details are described below.

\textbf{Step 1. Collect and integrate the data}:
As illustrated in Fig. 1, the AV operates on the road while continuously collecting multimodal sensory data (e.g., LiDAR, camera, and radar feeds). Simultaneously, it processes the user’s operational requirements. Before data transmission, the AV aggregates and formats the data according to the specifications of embodied world models, preparing it for transfer to the RSU.

\textbf{Step 2. Construct an MLMF Stackelberg game}:
In the Stackelberg framework, RSUs first set bandwidth prices using historical demand patterns, while AVs then calculate bandwidth needs based on current pricing and select cost-performance-optimized RSUs to establish connections. This iterative bandwidth-demand adjustment establishes an MLMF Stackelberg game in VEANs.

\textbf{Step 3. Design a TMABLPPO algorithm to find optimal strategies}:
To find the optimal strategies for AVs' and RSUs' continuous adjustments in the Stackelberg game, we designed the TMABLPPO algorithm in which Bi-LSTM modules can enhance the dynamical resource allocation ability, so that we can find the Stackelberg equilibrium. TMABLPPO includes dual actor models for the AVs and the RSUs, respectively. The AVs and RSUs execute the specific actor model to generate bandwidth selling prices and bandwidth requests, respectively. Then, the AVs choose optimal RSUs to establish connections and migrate VEAATs to target RSUs.

\textbf{Step 4. Undertake embodied task and action planning}:
After establishing real-time connections, AVs transmit the vehicular embodied AI agents' data to RSUs. With this data, RSUs build VEAATs by their computation resources, storage resources, etc. The VEAATs use the embodied world model with spatial awareness and long-horizon extrapolation proficiencies to analyze information such as AVs’ sensory data and user requirements. Consequently, the VEAATs generate the optimal action plans for AVs. The few-shot learning dynamically weights cross-modal features in order to enhance the performance of the embodied world models and their robustness in different states.

\textbf{Step 5. Output action planning}:
VEAATs pack planning routes and human-machine interface feedback. RSUs transmit this data using real-time connections with AVs.

\textbf{Step 6. Execute the action}:
AVs receive action planning from RSUs, provide feedback to users, and execute these plans to achieve a seamless in-vehicle experience.

\subsection{Latency Model}

The sets of AVs and RSUs are denoted as $\mathcal{R}=\left\{1,\ldots,r,\ldots R\right\}$ and $\mathcal{V}=\left\{1,\ldots,v,\ldots V\right\}$, respectively. Furthermore, we consider that the RSUs have limited bandwidth and computation resources, and the AVs must select an appropriate RSU and a proper amount of bandwidth for data transfer. Accordingly, we set the transmission task as $J_{rv}=\left\{D_{rv}, T_v^{max}{,\alpha}_{rv}\right\}$, where $D_{rv}$ is the data size of the task, $T_v^{max}$ is the maximum delay tolerance for AV $v$, and $\alpha_{v}$ is the task importance of AV $v$ currently, which depends on its task type. For instance, emergency tasks (e.g., ambulances on duty and fire trucks on duty) are more important than normal tasks (e.g., calculating the routes and the smart cabin services). Given the strict latency and ultra-high reliability of the data transfer task, the 6G communications and the orthogonal time frequency space modulation technique are essential, which can provide outstanding performance in high-mobility scenarios \cite{shen2023five}, \cite{xiao2021overview}. We consider the latency of channel estimation, which is an important part of the orthogonal time frequency space modulation. We set the signal process speed of RSUs as $f_{signal}$ and the iterations of channel estimation as $k$ \cite{yuan2021data}. Hence, the latency of the process of channel estimation can be calculated as 
\begin{equation}
    T_{channel} = \frac{kD_{rv}}{f_{signal}}.
\end{equation}
Moreover, the implementation of multiple-input multiple-output technology facilitates the concurrent transmission of multiple signals while enabling the efficient allocation of bandwidth among multiple AVs. The amount of the bandwidth purchased by the AV from the RSU is $b_{rv}$, and the data transfer rate can be calculated as 
$r_{rv}=b_{rv}\log_2{\left(1+\frac{\rho h d_{rv}^{-\varepsilon_0}}{\sigma^2}\right)}.$
Here $\rho$ represents the transmitter power of the AV, $h$ represents the unit channel power gain, $d_{rv}$ is the distance between RSU $r$ and AV $v$, $\varepsilon_0$ represents the path-loss coefficient, and $\sigma^2$ is the additive white Gaussian noise power in the communication link \cite{shannon1948mathematical}. The delay of the data transfer task can be calculated by $T_{rv}=\frac{D_{rv}}{r_{rv}}$. Therefore, the total latency is
\begin{equation}
    T_{total} = T_{rv} + T_{channel}.
\end{equation}

\subsection{MLMF Stackelberg Game between AVs and RSUs}

\subsubsection{Utility models of AVs and RSUs}
Building upon the QoS-aware revenue framework, the RSU-AV pairing mechanism operates under a probabilistic decision-making framework \cite{zhang2023learning}, where AVs dynamically evaluate various indicators to select optimal RSU connections. The probability of pairing function between RSU $r$ and AV $v$ is formulated as

\begin{equation}
\theta_{rv}=\alpha_{v}\frac{\frac{1}{p_{r}}}{\sum_{l\in\mathcal{R}}\frac{1}{p_{l}}},
\end{equation}
where $\alpha_{v}$ is the task important parameter for AV $v$, $p_{r}$ is the bandwidth price determined by RSU $r$. From the AVs' perspective, the targets are to minimize the transmission latency and the bandwidth cost.

The utility function of the AV consists of two parts: 
\begin{enumerate}
    \item [i)] {The first part we consider is the revenue function of AV $v$. Based on the user experience analysis, we formulate the QoS-oriented revenue function that integrates bandwidth allocation efficiency and latency-sensitive operational constraints. By integrating the Weber-Fechner Law that characterizes the logarithmic nature of the human-centric service perception in service quality evaluation, we incorporate a logarithmic component into the revenue function \cite{fechner1860elemente}. Hence, we obtain the revenue function of AV $v$ with the QoS as 

    \begin{equation}
    f_{rev} = \beta ln (e + \frac{\alpha_{v} b_{rv}}{T_{v}^{max}} ),
    \end{equation}
    where $\beta$ is the marginal effect parameter for the human-centric service perception. We also incorporate the Euler number to ensure the strict positivity of  $f_{rev}$.}
    \item [ii)]  Additionally, the second part is the cost function of AV $v$. The cost function is formulated as $f_{cost} = p_{r}b_{rv}$, where $p_r$ is denoted as the bandwidth price of RSU $r$.
\end{enumerate}
Integrating the revenue function and the cost function, we design the utility function of AV $v$ with the social effect as

\begin{equation}
U^F_v=\sum_{r\in \mathcal{R}}\big[\alpha_{v}\frac{\frac{1}{p_{r}}}{\sum_{l\in\mathcal{R}}\frac{1}{p_{l}}}\beta ln (e + \frac{\alpha_{r} b_{rv}}{T_{v}^{max}} )-p_{r}b_{rv}\big].
\end{equation}

Relating to the AVs' strategies, the RSU's utility function is based on the bandwidth amount that AVs bought. According to the net profit function, we assume that RSU $r$ has its base cost, which is the resource reservation overhead required to maintain competitive service quality in the game, and is denoted as $c_r$. Hence, the utility of RSU $r$ is denoted as 

\begin{equation}
U^L_r=\sum_{v\in\mathcal{V}}\big[\alpha_v\frac{\frac{1}{p_{r}}}{\sum_{l\in\mathcal{R}}\frac{1}{p_{l}}}(b_{rv}p_{r}-b_{rv}c_{r})\big].
\end{equation}

For simplicity, all strategies of AVs and RSUs are represented as vectors $\boldsymbol{B}=\{b_v\}_{v\in\mathcal{V}}$ and $\boldsymbol{P}=\left\{p_r\right\}_{r\in\mathcal{R}}$, respectively. The strategies of AVs excluding AV $v$ and RSUs excluding RSU $r$ are expressed as
$\boldsymbol{B_{-v}}=\{b_{v'}\}_{v'\in\mathcal{V}{\backslash}v}$ and $\boldsymbol{P_{-r}}=\{p_{r'}\}_{r'\in\mathcal{R}{\backslash}r}$.

\subsubsection{Stackelberg game formulation}
The Stackelberg Game establishes a hierarchical decision-making framework for optimizing bandwidth resource allocation between AVs and RSUs. The RSUs acting as leaders can analyze the previous games and the current state to decide the bandwidth price policy. After RSUs offer the price information, AVs acting as followers analyze the latency of the task and the bandwidth price of each RSU to figure out the bandwidth amount they need. Once the price of the bandwidth provided by RSUs has been determined, the follower-level problem can be formulated as follows,

\begin{equation}
    \begin{split}
    \textbf{\textit{P1:}}\:&\max\limits\:U^F_v(b_v, \boldsymbol{B^*_{-v}}, \boldsymbol{P}),  \\
    &\:\:s.t.\:\: {b_{rv} \geq 0}, \sum_{j\in\mathcal{R}}{T_{rv} \leq T^{max}_i}.
    \end{split}
    \label{problem1}
\end{equation}

At the leader level, the RSU employs a pricing strategy to influence the amount of bandwidth purchased by AVs to maximize its net profit. Accordingly, the leader-level problem is formulated as follows
\begin{equation}
    \begin{split}
    \textbf{\textit{P2:}}\:&\max\limits\:U^L_r(p_r, \boldsymbol{P^*_{-r}}, \boldsymbol{B}),  \\
    &\:\:s.t.\:\: {p_{rv}\in[c_r, p^{max}]},
    \end{split}
    \label{problem2}
\end{equation}
where $p^{max}$ is the maximum selling price, which enhances the stability of interactions between the RSU and the AV, thereby preventing the RSU from taking undue risks. Consequently, \textbf{\textit{P}}\hyperref[problem1]{\textbf{\textit{1}}} and \textbf{\textit{P}}\hyperref[problem2]{\textbf{\textit{2}}} can be regarded as a unified Stackelberg game \cite{huang2022joint}. The purpose is to find the SE that will yield the perfect outcome for this framework. Within this equilibrium configuration, no unilateral deviation by either RSUs or AVs can Pareto-improve individual payoffs, as all agents operate at consistent best-response strategies. The equilibrium enforces mutual optimality where each agent's utility is maximized conditional on counterparties executing rational best-response strategies aligned with their self-interest. In light of the aforementioned, the SE of our model can be interpreted as the subsequent definition.

\begin{definition}
(SE): \textit{Initially, optimal price strategies of RSUs and optimal bandwidth strategies of AVs are set as $\boldsymbol{p^\ast}=\left\{p_r^\ast\right\}_{r\in\boldsymbol{R}}$ and $\boldsymbol{b^\ast}=\left\{b_v^\ast\right\}_{v\in\boldsymbol{V}}$, respectively. The optimal strategy functions of all other RSUs and AVs except $r$ and $v$ are denoted as $\boldsymbol{B_{-v}^\ast\ }=\left\{b_{v’}^\ast\right\}_{v'\in\mathcal{V}{\backslash}v}$ and $\boldsymbol{P_{-r}^\ast}=\left\{p_{r’}^\ast\right\}_{r'\in\mathcal{R}{\backslash}r}$. Consequently, we can establish a potential stable point of a dynamic adjustment process in which individuals adjust their behavior to that of the other players in the game, searching for strategy choices that will yield superior results. Subsequently, the stable point $(\boldsymbol{b^\ast},\boldsymbol{p^\ast})$ is defined as the SE, which satisfies the following inequalities
\begin{equation}
\begin{cases}
U_v^F\big({b_v^*},\boldsymbol{B_{-v}^*},\boldsymbol{p^*}\big)\geq U_v^F\big({b_v},\boldsymbol{B_{-v}^*},\boldsymbol{p^*}\big),&\forall v\in\mathcal{V}, \\
U_r^L\big({p_r^*},\boldsymbol{P_{-r}^*},\boldsymbol{b^*}\big)\geq U_r^L({p_r},\boldsymbol{P_{-r}^*},\boldsymbol{b^*}),&\forall r\in\mathcal{R}.
\end{cases}
\end{equation}}
\end{definition}

\begin{figure*} 
    \centering
    \includegraphics[width=0.7\linewidth]{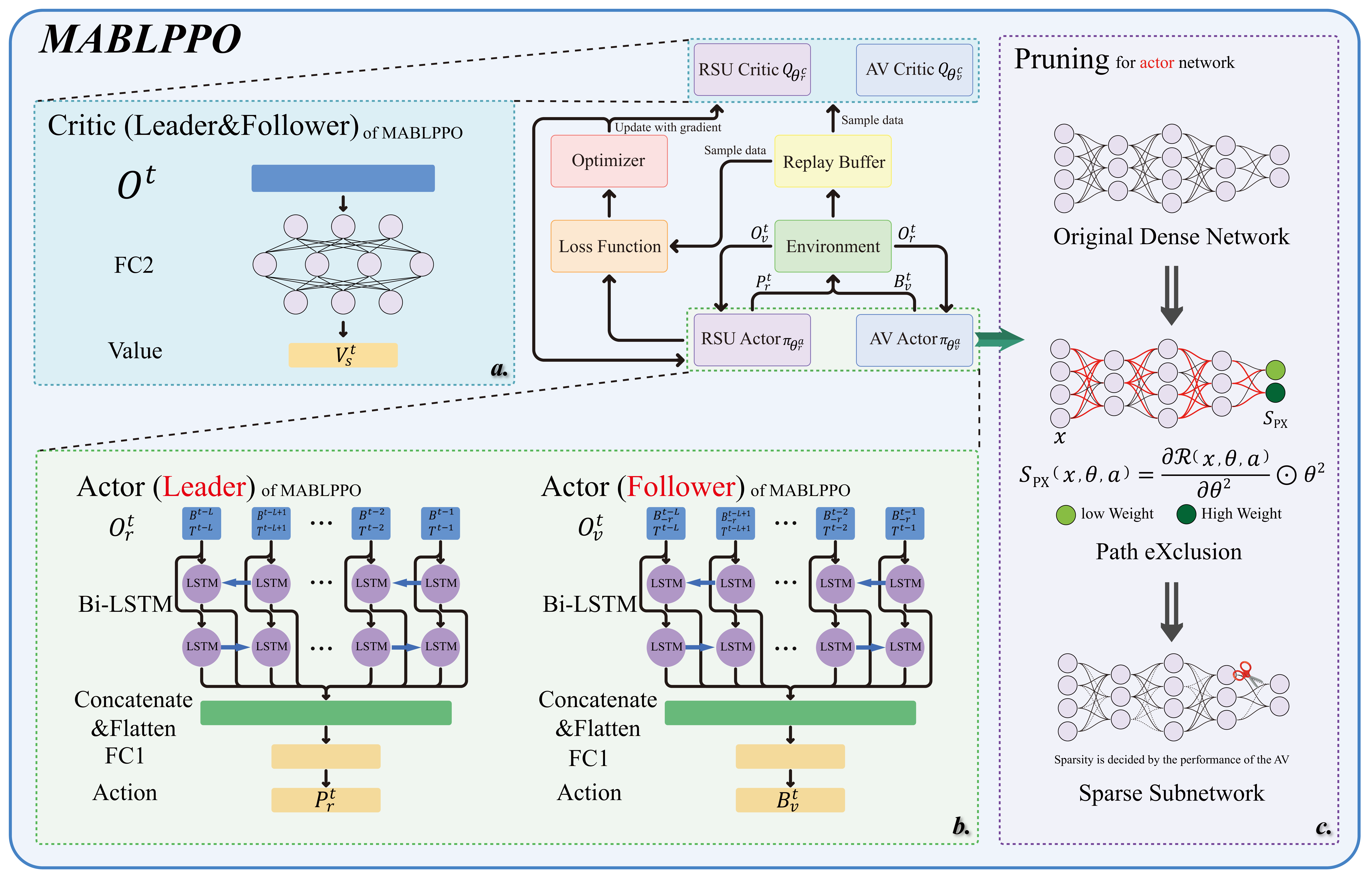}
    \caption{TMABLPPO algorithm's Framework for the VEAAT migration.}
    \label{lstm}
\end{figure*}

\section{Equilibrium Analysis for MLMF Stackelberg Game} \label{equilibrium_analysis}

\subsection{Follower-level Equilibrium Analysis}
At the follower level of the Stackelberg game, each AV $v$ obtains the bandwidth it needs based on the information of tasks and RSUs’ bandwidth prices \cite{zhou2022stackelberg}. To analyze the concavity of the utility function for AV $v$, the first-order and second-order derivatives of $U^F_v$ with respect to $b_{rv}$ are expressed as
\begin{equation}
    \frac{\partial U^{F}_{v}}{\partial b_{v}} =\sum_{r \in \mathcal{R}}\big[\alpha_{v} \beta \frac{\frac{1}{p_{r} } }{\sum _{l\in\mathcal{R}}\frac{1}{p_{l} } }(\frac{\frac{\alpha _{v} }{T^{max} _{v} } }{ e+\frac{\alpha_{v} b_{rv} }{T^{max} _{v} }}-p_{r})\big],
\end{equation}

\begin{equation}
    \frac{\partial^2 U^F_{v}}{\partial b^2_{v}}=\sum_{r \in \mathcal{R}}\big[-\alpha_v \beta \frac{\frac{1}{p_r}}{\sum_{l \in \mathcal{R}\frac{1}{p_l}}}\frac{{(\frac{\alpha_v}{T^{max}_v})^2}}{(e+\frac{\alpha_v b_{rv}}{T^{max}_v})^2}\big]< 0.
\end{equation}

Because the second-order derivative of the follower utility function is negative, $U^F_v$ is quasi-concave in $b_{rv}$. Hence, the maximum exists for the first-order derivative to equal zero. Hence, the maximum of the follower’s utility function exists and satisfies the first-order optimality condition $\frac{\partial U^F_v}{\partial b_{rv}}=0$,
\begin{equation} \label{response_follwers}
    \hat{b_{rv}}=\frac{1}{p_r} -\frac{eT^{max}_{v}}{\alpha_v},
\end{equation}
where $\hat{b_{rv}}$ represents AV $v$ unique optimal strategy that is calculated by the first-order derivative. The bandwidth amount $b_{rv}$ follows the constraint $b_{rv} \in [0, b^{max}_{rv}]$. Hence, the optimal strategy for AV $v$ is
\begin{equation}
    b^\star_{rv}=\begin{cases}
    \hat{b_{rv}}, &\mathrm{if} \ \alpha_v\geq ep_rT^{max}_v,\\
    0, &\mathrm{if} \ \alpha_v < ep_rT^{max}_v . 
    \end{cases} 
\end{equation}

\subsection{Leader-level Equilibrium Analysis}
In a Stackelberg game, leaders and followers interact such that the leaders’ optimal strategies are influenced by the followers’ actions. Consequently, analyzing the leaders’ strategies requires incorporating the followers’ best-response strategies, i.e., substituting Eq. (\ref{response_follwers}) into Eq. (\ref{utility_leaders}). 

Consequently, analyzing the leaders' strategies requires incorporating the followers' optimal strategies. Thus, we embed the followers' optimal strategies from Eq. (\ref{response_follwers}) into leaders' strategies. We denoted $y_r=\frac{1}{p_r}$ to simplify the mathematical equation. Hence, the utility function of leader $r$ is 
\begin{equation} \label{utility_leaders}
    \begin{split}
    U^{L}_r &= \sum _{v\in \mathcal{V} }\big(\alpha_v\frac{\frac{1}{p_r} }{\sum _{l\in \mathcal{R}\frac{1}{p_l} }}(b_{rv}p_r-b_{rv}c_r)\big), \\ 
    &= \sum _{v\in \mathcal{V} }\big(\alpha_v\frac{y_r}{\sum _{l\in \mathcal{R}}y_l}(y_r -\frac{eT^{max}_{v}}{\alpha_v})(\frac{1}{y_r}-c_r)\big).
    \end{split}
\end{equation}

\begin{lemma}
A function $\mathcal{H}_r(\boldsymbol{B})$ is a standard function if and only if it satisfies the following three conditions:
\begin{itemize}
\item Positiveness: $\mathcal{H}_r(\boldsymbol{B})>0$.
\item Monotonicity: $\forall \boldsymbol{B'} > \boldsymbol{B}, \mathcal{H}_r( \boldsymbol{B'})> \mathcal{H}_r(\boldsymbol{B})$.
\item Scalability: $\forall \lambda > 1, \lambda \mathcal{H}_r(\boldsymbol{B})> \mathcal{H}_r(\lambda \boldsymbol{B})$.
\end{itemize}
where $\mathcal{H}_r(\boldsymbol{B})$ denotes the optimal strategy of RSU $r$.
\end{lemma}

\begin{theorem} \label{theorem1}
The unique SE denoted as $(\boldsymbol{B^\star}, \boldsymbol{P^\star})$ is established in the formulated game when the following condition holds, i.e., $\alpha_v\geq ep_rT^{max}_v$. In this case, AVs’ bandwidth amount strategies and RSUs’ bandwidth pricing strategies are both optimal. In particular, this equilibrium ensures the simultaneous optimization in VEANs.
\end{theorem}

\begin{proof} 

We have incorporated the followers’ optimal strategies into the leaders’ utility function as described in Eq. (\ref{utility_leaders}). To streamline the following analysis, we denote $W = \sum _{l \in \mathcal{R}/r} y_l$. Then, we analyze the first-order and second-order derivatives of $U^L_r$ with respect to $y_r$,


\begin{equation}
    \frac{\partial U^L_r }{\partial y_r } =\sum_{v \in \mathcal{V}}\alpha_v\frac{-c_r y_r^2-2c_ry_rW +\frac{eT^{max}_v}{\alpha_v}(c_rW+1)+W }{{( \sum _{l \in \mathcal{R}} y_l ) ^2 }},
\end{equation}


\begin{equation}
\begin{aligned}
    \frac{\partial^2 U^L_r }{\partial y^2_r } =-2\sum_{v \in \mathcal{V}}\big(\alpha_v\frac{c_rW^2+\frac{eT^{max}_v}{\alpha _v}(c_rW + 1)+W}{{(\sum _{l \in \mathcal{R}} y_l)^3 }}\big)<0 .
\end{aligned}
\end{equation}
Because the leader utility function is quasi-concave in $b_{rv}$, the maximum exists where the first-order derivative equals zero. To brief the equations, we denote $Z_v = \frac{eT^{max}_v}{\alpha_v}$. If the maximum of the leader's utility function exists and satisfies the first-order optimality condition $\frac{\partial U^L_r}{\partial y_r}=0$, then we can get
\begin{equation}
    \hat{y_{j}}=\sum_{v \in \mathcal{V}}\big(\frac{\sqrt{c_r^2W^2+c_r(Z_v+(1+Z_v c_r )W)}}{c_r}-W\big) >0.
\end{equation}

Because of the constraint $p_{j} \in [c_r, p^{max}]$, the optimal strategy for RSU $r$ is
\begin{equation}
\begin{aligned}
    y^\star_{r}&=\mathcal{G}(\boldsymbol{Y}) \\
    &=\begin{cases}
    \hat{y_{r}}, &\mathrm{if}\: c_rW<\sqrt{c_r^2W^2+c_r(Z_v+(1+Z_v c_r )W}, \\
    0, &\mathrm{if}\: c_rW\geq\sqrt{c_r^2W^2+c_r\big(Z_v+(1+Z_v c_r )W\big)}. 
    \end{cases} 
\end{aligned}
\end{equation}

The leader-level game has a unique Nash Equilibrium if the best response function of RSUs satisfies the standard function properties \cite{xu2021privacy}. Hence, we will prove $\mathcal{G}(\boldsymbol{Y})=y^\star_j$ is the standard function.
\subsubsection{Positiveness}
    It’s brief to prove that the $\mathcal{G}(\boldsymbol{Y})=y^\star_j$ is a positive function.
\subsubsection{Monotonicity}
It’s easy to understand this function $\frac{\partial W}{\partial \boldsymbol{Y}} >0$. Hence, we use the chain rule in the derivatives in two parts: $\frac{\partial \mathcal{G}(\boldsymbol{Y})}{\partial \boldsymbol{Y}} =\frac{\partial \mathcal{G}(\boldsymbol{Y})}{\partial W} \frac{\partial W}{\partial \boldsymbol{Y}}$. The first derivative of $\mathcal{G}(\boldsymbol{Y})$ in terms of $W$ is
\begin{equation}
    \frac{\partial \mathcal{G}(\boldsymbol{Y})}{\partial W}=\sum _{v\in\mathcal{V}}[\frac{2C_rW+1+Z_vc_r}{\sqrt{(2c_rW+1+Z_vc_r)^2-(c_rW-1)^2} } -1]>0.
\end{equation}
Since the first-order derivative of $y^\star_j$ is positive, $y^\star_j$ is a monotonic function.

\subsubsection{Scalability}
We set a parameter denoted as $\beta$. The constraint of it is $\beta \in (1,+\infty)$. $\mathcal{F}_{sca} = \beta\mathcal{G}(\boldsymbol{Y})-\mathcal{G}(\beta\boldsymbol{Y})$ is formulated as

\begin{equation}
\begin{aligned}
    \mathcal{F}_{sca} =&\sum_{v \in \mathcal{V}}\big(\frac{\beta\sqrt{c_r^2W^2+c_r\big(Z_v+(1+Z_v c_r )W\big)} }{c_r} \\
    &-\frac{\sqrt{{c_rW\beta}^2+c_r(Z_v+(1+Z_v c_r )W \beta)} }{c_r}\big) > 0.
\end{aligned}
\end{equation}

\end{proof}

In summary, $\mathcal{G}(\boldsymbol{Y})=y^\star_j$ is a standard function that satisfies the equation $p^\star_j=\frac{1}{y^\star_j}$. Hence, $p^\star_j$ is an optimal response function. The leader-level response function can find the value that maximizes the utility function of leaders. According to Theorem \ref{theorem1}, the mathematical models of leaders and followers can formulate an MLMF Stackelberg game with a unique SE.

\section{Bi-LSTM Based Multi-agent Deep Reinforcement Learning Algorithms with Pruning} \label{MABLPPO_with_pruning}

In the context of the highly intricate data-transferring environment, the DRL algorithm is better positioned to leverage past experience to inform decision-making and rapidly identify a game equilibrium solution than greedy algorithms \cite{kang2024metaverses}. In the pursuit of privacy protection in real-world settings, decentralized algorithms such as DRL can integrate the observation of partial information by cooperative agents in actual environments. Consequently, we transform the model into a partially observable Markov decision process. The framework of TMABLPPO algorithm is illustrated in Fig. \ref{lstm}.

\subsection{Bi-LSTM Based Actor Algorithms}
The Bi-LSTM model exhibits enhanced processing capabilities with respect to contextual information. In the context of DRL, the model is designed to accept the state of previous steps as input, thereby facilitating more effective processing of the observed information. Furthermore, the Bi-LSTM is capable of capturing the inter-agent dependencies within the same observation space, thereby facilitating the generation of more rational actions by each agent. The Bi-LSTM model is based on the traditional LSTM model framework, with the addition of a bidirectional propagation mechanism. The LSTM model is composed of three gates: the input gate, the output gate, and the forget gate. In a Bi-LSTM module, two LSTM chains are employed \cite{huang2015bidirectional}. Each processing a forward or backward sequence, as defined
\begin{equation}
h_{t+1}^{fwd}=LSTM_{forward}(x_t, h_{t}^{fwd}),
\end{equation}
\begin{equation}
h_t^{bwd}=LSTM_{backward}(x_t, h_{t+1}^{bwd}),
\end{equation}
where $h_t^{fwd}$ denotes the hidden state of forward LSTM, $_t^{bwd}$ denotes the hidden state of backward LSTM and $x_t$ denotes the input vector at timestep $t$.

The initial MLP processes the LSTM layer's outputs by transforming high-dimensional sequence features into a decision-optimized feature space through non-linear activation functions like ReLU, which captures intricate data patterns. Subsequent stacked MLP layers progressively extract higher-level abstract features by iteratively refining preceding layer outputs. This enables the network to model complex input-action mappings and enhance the agent capacity to select optimal actions in dynamic states.

\subsection{Multi-agent DRL Algorithms for Stackelberg Game}
In an optimal setting, if the RSU and the AV can access global information, then they will make optimal decisions. However, in light of the growing significance of data privacy concerns, the implementation of robust privacy protection mechanisms has emerged as a pivotal aspect within the domain of artificial intelligence. Accordingly, the model is transformed into a partially-observable DRL, which is defined as follows:
\begin{enumerate}
    \item \textbf{Observation}: Both the RSU and the AVs are privy to only a portion of the real-time information and the historical information of the previous $L$ rounds. In the current time slot, the observation information available to RSU $r$ is the historical price strategy and demand strategy $o_r^t\triangleq\{B^{t-L},P^{t-L},...,B^{t-1},P^{t-1}\}$, which can be used to inform the formulation of a strategy $p_j^t$. Subsequently, AVs $v$ are able to observe the current pricing of the RSU, as well as the historical price strategy and demand strategy $o_v^t\triangleq\{B_{-v}^{t-L},P^{t-L},...,B_{-v}^t,P^t\}$, which enables them to formulate a strategy $b_v^t$. In the case when $t<L$, the absent data are substituted with $B^0$ and $P^0$. To brief expression, we integrate $o^t_r$ and $o^t_v$ as $o^t_k,k \in \mathcal{R} \cup \mathcal{V}$.
    \item \textbf{Action}: At time slot $t$, the RSU generates $p_j^t\in\left[c_j,p^{max}\right]$, according to the information $o_j^t$, where $p^{max}$ is the maximum value of a pricing hyperparameter and $c_j$ is the time spent by all data transferring. Furthermore, AVs generate $b_i^t\in[0,+\infty)$ according to the information $O_i^t$. To brief expression, we integrate $p^t$ and $b^t$ as $a^t_k,k \in \mathcal{R} \cup \mathcal{V}$.
    \item \textbf{Reward}: The RSUs and AVs interact with the environment based on their strategies, generating rewards that combine immediate operational outcomes and long-term sustainability metrics, enabling iterative optimization of cooperative vehicular intelligence.
\end{enumerate}

The MABLPPO algorithm is a variant of the multi-agent based Actor-Critic framework. This approach offers a more effective means of addressing the inherent challenges of collaboration and competition between RSUs and AVs while facilitating the attainment of an optimal equilibrium solution. The actor network is responsible for generating an optimal strategy in the current state based on the current observed and previous information. The loss function of the actor model in MABLPPO is formulated as 
    \begin{equation}\label{lossFunction}
        \begin{split}
        L_{\mathrm{actor}}(\theta^k)=\mathbb{E}  \bigg[\min  \bigg(\frac{\pi_{\theta^k}(a^k|s^k)}{\pi_{\theta^k}^{old}(a^k|s^k)}\hat{A}_{\pi_{\theta^k}}(a^k,p^k) \\
        \text{clip}(\frac{\pi_{\theta^k}(a^k|s^k)}{\pi_{\theta^k}^{old}(a^k|s^k)},1-\epsilon,1+\epsilon)\hat{A}_{\pi_{\theta^k}}(a^k,p^k)\bigg) \bigg].
        \end{split}
    \end{equation}
    
The importance sampling ratio $\frac{\pi_{\theta^k}(a^k|s^k)}{\pi_{\theta^k}^{old}(p^k|s^k)}$ is a measure of the relative change in actions generated by the new policies $\pi_{\theta^k}(a^k|s^k)$ and old policies $\pi_{\theta^k}^{old}(a^k|s^k)$ in the current state for the agent $k$. ${\hat{A}}_{\pi_{\theta k}}(s^k,a^k)$ is an advantage function that is used to motivate policy updates.

The Critic network uses the Temporal Difference (TD) error to measure the difference between the current and expected states. This mechanism is an essential part of the Actor-Critic algorithmic framework. The TD error function is defined as 
    \begin{equation}
        d^k=r^k(t)+\gamma V_{\omega^k}(s^k(t+1))-V_{\omega^k}\big(s^k(t)\big),
    \end{equation}
where $d^k$ denotes the TD error of the agent $k$, $r^k(t)$ denotes the current reward of the agent $k$ at timestep $t$, $\gamma$ denotes the discount factor and $V_{\omega^k}$ denotes the value function of agent $k$. 
The loss function of the critic network in MAPPO is 
    \begin{equation}
        L_{critic}(\omega^k)=\min\mathbb{E}\left[\left(d^k\right)^2\right].
    \end{equation}

\subsection{Computation-aware Pruning Algorithm for Efficient MADRL}
In the context of a connected vehicle environment, where computing resources are constrained and heterogeneous and where latency requirements are exacting, personalized pruning can be employed to reduce the number of parameters in a model and thereby reduce the time taken for an agent to make its strategy. Personalized pruning provides suitable models for different computing platforms and includes two steps: 
\begin{enumerate}
    \item Evaluating the computation capability of the vehicle to choose the optimal pruning rate.
    \item Using the PX algorithm \cite{iurada2024finding} to prune the actor network with optimal density.
\end{enumerate}
After pruning the actor network, the inference speed is accelerated, and the computation resource demand is lower.

\subsubsection{Path eXclusion algorithm}
This section presents an introduction to the PX pruning algorithm. Firstly, the agent network is defined as
\begin{equation}
    \mathcal{Q}:\mathbb{R}^d\rightarrow \mathbb{R}^k,
\end{equation}
where $\mathbb{R}^d$ denotes the input vector of the agent network and $\mathbb{R}^k$ denotes the output vector of the agent network.

Unlike traditional supervised learning, reinforcement learning generates training data through dynamic interactions between an optimized policy model and its environment. Hence, we use the optimal policy to generate a dataset of $N$ data points for the pruning algorithm, in which agent observations serve as input features and corresponding actions constitute output labels. The dataset is denoted as $(X_{Obs}, Y_{Act}) = \{(x_i, y_i)\}^N_{i=1}$, which $x_i \in \mathbb{R}^d$ and $y_i \in \mathbb{R}^k$. The problem of an unstructured network pruning can be formulated as a binary mask $M\in\{0,1^m\}$.

\begin{equation}
    \begin{split}
    \textbf{\textit{P3:}}\:&\min_M\frac1N\sum_{i=1}^N\mathcal{L}\big(\mathcal{Q}(\boldsymbol{x}_i;\mathcal{A}(\boldsymbol{\theta},\boldsymbol{M})\odot \boldsymbol{M}),\boldsymbol{y}_i\big) \\
    &\:\: {\mathrm{s.t.~}\boldsymbol{M}\in\{0,1\}^m,\parallel \boldsymbol{M}\parallel_0/m\leq1-q},
    \end{split}
\end{equation}
where $\mathcal{L}$ is the loss function for the downstream task, $q$ denotes the desired sparsity for the target network, $\mathcal{A}$ denotes the actor network algorithm, $\boldsymbol{\theta}$ is the parameter of the actor network and $\odot$ denotes the Hadamard product.
	
Once the score has been calculated, only the top-$S$ mask elements will be retained for use in accordance with the aforementioned formula. The saliency function is defined as
\begin{equation}
    S_{PX}(\boldsymbol{x},\boldsymbol{\theta},\boldsymbol{a})=\frac{\partial R(\boldsymbol{x},\boldsymbol{\theta},\boldsymbol{a})}{\partial\boldsymbol{\theta}^2}\odot\boldsymbol{\theta}^2.
\end{equation}

The saliency function is now subjected to analysis. Initially, the model output is optimized through the utilization of a first-order Taylor expansion, thereby approximating the model output at the preceding time step $t$ is formulated as 
\begin{equation}
    \mathcal{Q}(\boldsymbol{x},\boldsymbol{\theta}_{t+1})=\mathcal{Q}(\boldsymbol{x},\boldsymbol{\theta}_t)-\alpha\Theta_t(\boldsymbol{x},\boldsymbol{x})\nabla_{\mathcal{Q}}\mathcal{L},
\end{equation}
where $\Theta_t(\boldsymbol{x},\boldsymbol{x})=\nabla_\theta \mathcal{Q}(\boldsymbol{x},\boldsymbol{\theta}_t)\nabla_\theta \mathcal{Q}(\boldsymbol{x},\boldsymbol{\theta}_t)^T\in\mathbb{R}^{NK\times N K}$ is the Neural Tangent Kernel (NTK) in the time slot $t$.

Recent studies demonstrate that the NTK captures the training dynamics of deep neural networks across frameworks, with the eigenvector corresponding to its largest eigenvalue critically influencing model convergence rates \cite{liu2024ntk}. This intrinsic property positions NTK-derived metrics as powerful theoretical tools for quantitatively evaluating model performance, particularly in comparing architectural advantages. Leveraging NTK's mathematical universality and empirical versatility, it provides a robust method for computing node importance weights in the MAPPO algorithm's actor networks in decision-making systems.

We consider that all network paths connect input and output neurons $\mathcal{P}=\{1,\ldots,p,\ldots,P\}$. The presence of weight $\boldsymbol{\theta}_i$ in a path is $p$ indicated by the symbol $p_i=\mathbb{I}[\boldsymbol{\theta}_i\in p]$, and the product of the weights of a path is $\upsilon_p(\boldsymbol{\theta})=\prod_{i=1}^{m}\boldsymbol{\theta}_i^{p_i}$. When inputting $\boldsymbol{x}\in X$, the activation status of a path is given by the expression $a_p(\boldsymbol{x},\theta)=\prod_{{i|\theta_i\in p}}\mathbb{I}[z_i>0]$, where is the activation of the neuron $z_i$ connected to the previous layer through $\theta_i$. Consequently, the output function of the first layer can be expressed as
\begin{equation}
    \mathcal{Q}^k(\boldsymbol{x},\boldsymbol{\theta})=\sum_{s=1}^{d}\sum_{p\in\mathcal{P}_{s\rightarrow k}} v_p(\boldsymbol{\theta})a_p(\boldsymbol{x},\boldsymbol{\theta})\boldsymbol{x}_s,
\end{equation}
where $\boldsymbol{x}_s$ denotes the $s$-th term of the $\boldsymbol{x}$ vector and $\mathcal{P}_{s\rightarrow k}$ is the set of all paths from input $s$ to output neuron $k$. 

Factorize NTK according to the matrix chain rule, which can be calculated as
\begin{equation}
\begin{aligned}
    \Theta(X,X)& =\nabla_{\boldsymbol{\theta}} \mathcal{Q}(X,{\boldsymbol{\theta}})\nabla_{\boldsymbol{\theta}} \mathcal{Q}(X,\boldsymbol{\theta})^T, \\
    &=J_v^\mathcal{Q}(X)\Pi_{\boldsymbol{\theta}}(J_v^\mathcal{Q}(X))^T,
\end{aligned}
\end{equation}
where $J_v^\mathcal{Q}(X)\in\mathbb{R}^{NK\times P}$ is called Path Activation Matrix. According to the Forbenius norm, we have rewritten NTK as
\begin{equation}
\begin{aligned}
    Tr[\Theta(X,X)]& =Tr[\nabla_ {\boldsymbol{\theta}} \mathcal{Q}(X,{\boldsymbol{\theta}})\nabla_{\boldsymbol{\theta}} \mathcal{Q}(X,{\boldsymbol{\theta}})^T], \\
    &\leq\parallel J_v^\mathcal{Q}(X)\parallel_F^2\cdot\parallel J_{\boldsymbol{\theta}}^v\parallel^2,
\end{aligned}
\end{equation}
where$\parallel J_{\boldsymbol{\theta}}^v\parallel_F^2=\sum_{p=1}^{P}\sum_{j=1}^{m}\left(\frac{v_p\left({\boldsymbol{\theta}}\right)}{{\boldsymbol{\theta}}_j}\right)^2$ and $J_{\boldsymbol{\theta}}^v\in\mathbb{R}^{P\times m}$ is called Path Kernel Matrix.

Two additional models, exhibiting the same structural characteristics as the primary model, are then employed to accelerate the computation process and enhance the parallelism of the calculation, which reduces the latency of pruning.
\begin{equation}
    h^k(1,{\boldsymbol{\theta}}^2,1)=\sum_{p=1}^{P} v_p({\boldsymbol{\theta}}^2)=\sum_{p=1}^{P} v_p^2({\boldsymbol{\theta}}),
\end{equation}
\begin{equation}
    g^k(\boldsymbol{x}^2,1,a)=\sum_{p=1}^{P} a_p(\boldsymbol{x},{\boldsymbol{\theta}})\boldsymbol{x}_{s|s\in p}^2.
\end{equation}

The variable “$1$” refers to an all-one vector. The outputs of the two networks
$\mathcal{R}(\boldsymbol{x},{\boldsymbol{\theta}},a)=\sum_{n=1}^{N}\sum_{k=1}^{K} g^k(\boldsymbol{x}_n^2,1,a_n) h^k(1,{\boldsymbol{\theta}}^2,1)$. Thus, the saliency function can be calculated in parallel, which is denoted as 

\begin{equation}
    S_{\mathrm{PX}}(\boldsymbol{x},{\boldsymbol{\theta}},a)= \parallel J_v^\mathcal{Q}(X)\parallel_F^2\cdot\parallel J_{\boldsymbol{\theta}}^\mathcal{Q}\parallel_F^2\odot{\boldsymbol{\theta}}^2.
\end{equation}

\subsubsection{Computation awareness}
The main idea of this section is to introduce a method for AVs with dynamic adaptation of computing resources. It is dependent on the GPU benchmark because the TMABLPPO model costs lots of GPU resources. The latency of the model relied not only on the AI Trillions or Tera Operations per Second (TOPS) but also on the GPU memory. Thus, we divided three performance classes for the AVs. It can avoid the problem of all AVs using the same model, as those with limited computing resources may cause longer processing delays when executing the model. To address this challenge, we use a computation-aware pruning algorithm, which can prune intelligently for different kinds of AVs. This is achieved so that each AV can run the suitable model according to its computing resources.

To address meaningful performance disparities among automotive cockpit platforms, we employ a tiered classification system relative to TOPS computation capability, defining $\mathcal{K}$ hierarchical intervals: $\{(0, C_1), [C_1, C_2),..., [C_{k-1}, C_k),..., [C_K, +\infty)\}$, which $\mathcal{K}=\{1, ..., k, ... K\}$. This system enables dynamic reconfiguration of neural network pruning rates across various computation tiers, ensuring balanced optimization between system efficiency and processing latency. The adaptive module integrated employs performance monitoring to compute the optimal pruning ratios automatically, ensuring systems remain responsive while utilizing all available computation resources within each platform's respective capability range. This tier-aware optimization feature optimizes the existing gaps in performance disparity by adjusting model compression techniques to hardware capacity limits. The detailed algorithm of TMABLPPO is present in \textbf{Algorithm }\hyperref[algorithm1]{\textbf{1}}.

\begin{algorithm}[t]
\caption{TMABLPPO-based Solution for MLMF Stackelberg Game}
\label{algorithm1}
\begin{algorithmic}
    \State Initialize maximum price $P_{max}$, maximum bandwidth $B_{max}$, batch size $bs$, maximum episodes $E$, maximum time slots $T$ in one episode;
    \For{Agent $k \in \mathcal{R} \cup \mathcal{V}$}
        \State Initialize actor $\pi_{\boldsymbol{\theta}_k^a}$, $\pi_{\boldsymbol{\theta}_k^a}^{old}$, critic $Q_{\boldsymbol{\theta}_k^c}$;
    \EndFor
    \For{Episode $1, 2, ... ,E$}
        \State Reset Stackelberg game environment state $S_0$ and clear 
        \State up the replay buffer $\mathcal{B}$;
        \For {Time slot $t = 1, 2, ... ,T$}
            \State Input $o^t_r$ into RSU $r$ actor policy $\pi_{\boldsymbol{\theta}_r^a}$;
            \State Infer the current price strategy $p^t_r$;
            \State Input $o^t_v$ into AV $v$ actor policy $\pi_{\boldsymbol{\theta}_v^a}$;
            \State Infer the current bandwidth amount strategy $b^t_v$;
            \State Update $S_t$ to $S_{t+1}$
            \State Calculate rewards $R^t_r$ for RSU $r$ and $R^t_v$ for AV $v$;
            \State Store transition $(o^t_r,o^t_v,R^t_r,R^t_v)$ into $\mathcal{B}$;
            \If {$t \% |bs| == 0$}
                \State Sample batch data from $\mathcal{B}$;
                \State Process batch data with advantage estimation;
                \State Calculate the loss using Eq. (\ref{lossFunction});
                \State Update $\theta_k^a$, $\theta_k^c$ using the loss;
            \EndIf
        \EndFor
    \EndFor
    \For{Agent $v \in \mathcal{V}$}
        \State Evaluate the computation resources of AV $v$;
        \State Initialize the pruning rate of the actor model $pr$;
        \State Find the optimal pruning rate $pr$ according to computation resources;
        \State Prune the actor model $\pi_{\boldsymbol{\theta}_v^a}$ with $pr$;
        \State Finetune the actor model $\pi_{\boldsymbol{\theta}_v^a}$;
    \EndFor

\end{algorithmic}
\end{algorithm}

\begin{figure*}
    \centering
    \begin{minipage}{0.32\textwidth}
        \centering
        \includegraphics[width=\linewidth]{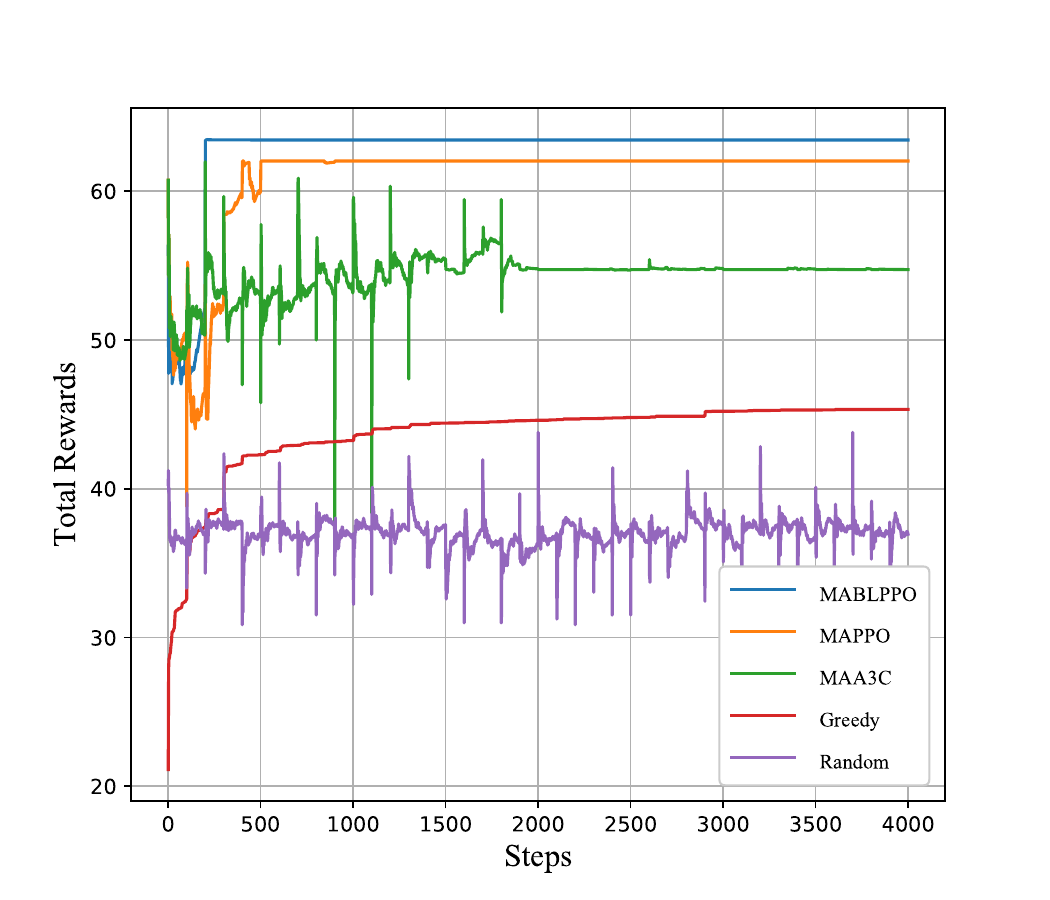}
        \caption{Comparison of the total reward curves of MABLPPO and baselines for the Stackelberg Game.}
        \label{before_prune}
    \end{minipage}
    \hfill
    \begin{minipage}{0.32\textwidth}
        \centering
        \includegraphics[width=\linewidth]{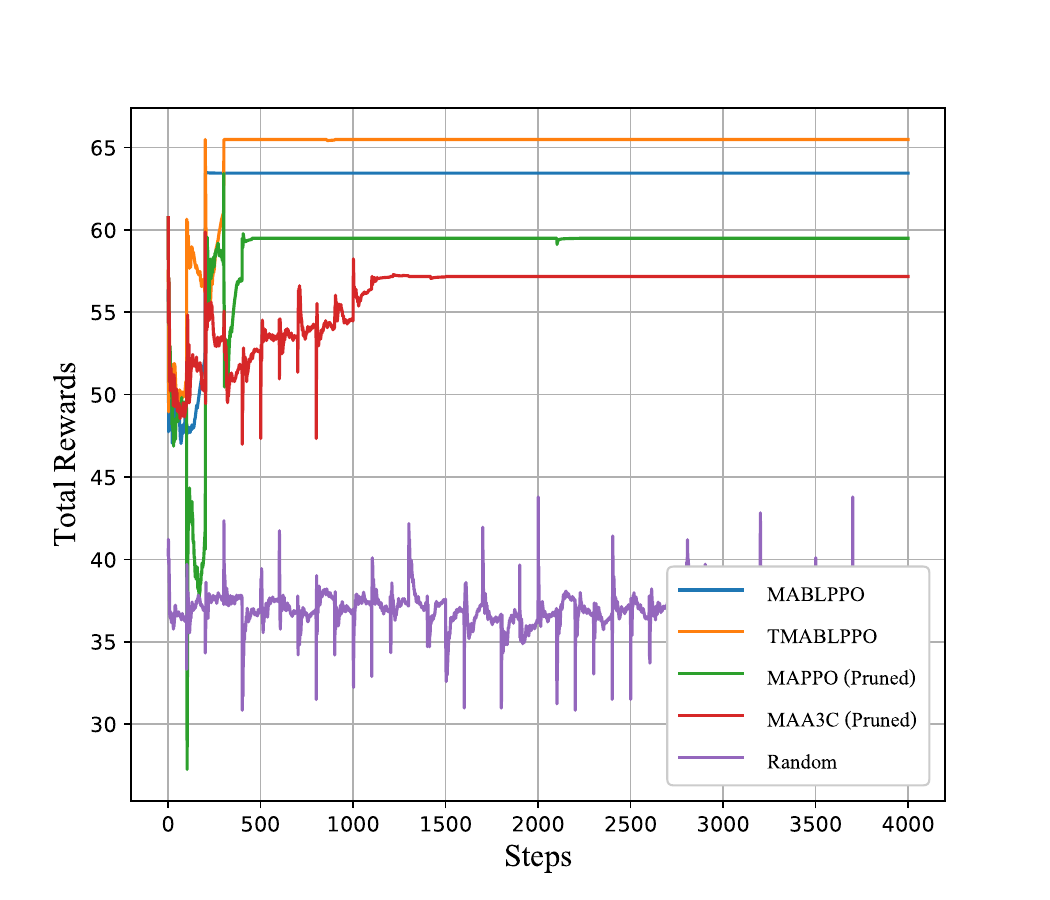}
        \caption{Comparison of the total reward curves of TMABLPPO and baselines for the Stackelberg Game with 90\% density.}
        \label{after_prune}
    \end{minipage}
    \hfill
    \begin{minipage}{0.32\textwidth}
        \centering
        \includegraphics[width=\linewidth]{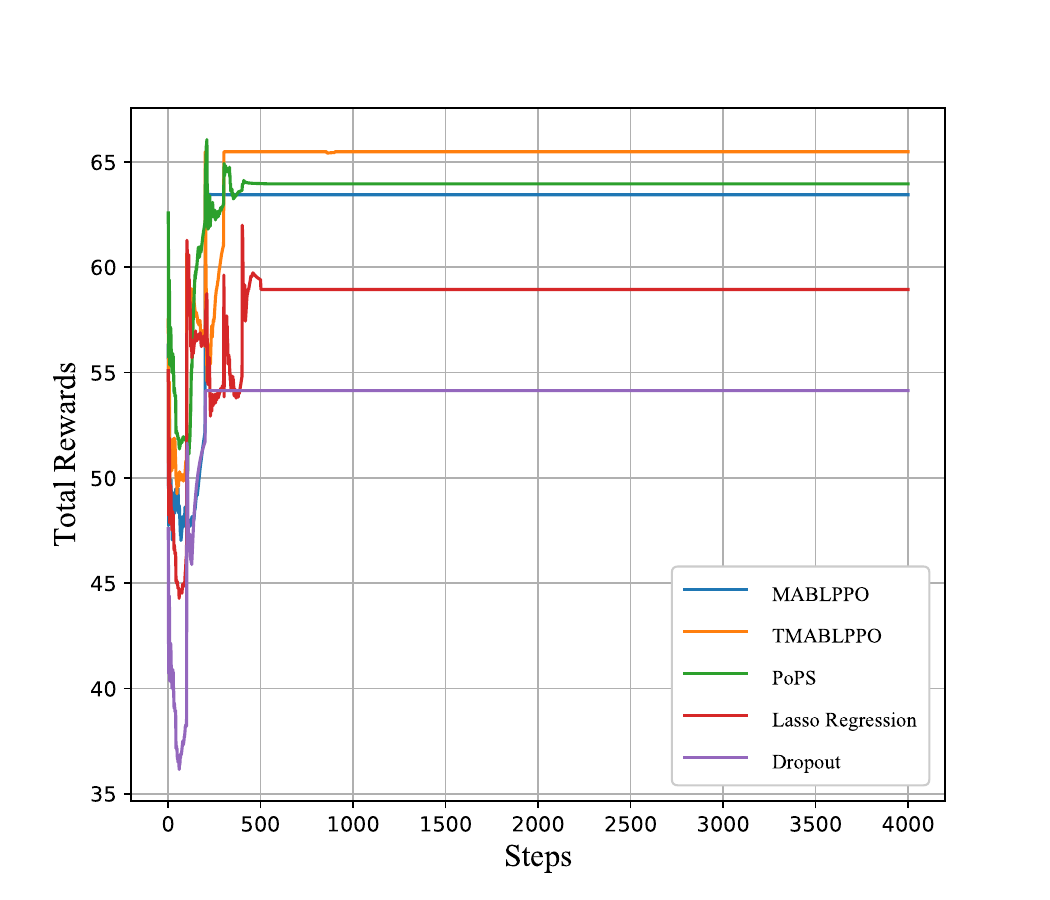}
        \caption{Comparison of the total reward curves of TMABLPPO and other pruning algorithms for the Stackelberg Game with 90\% density.}
        \label{compare_prune_algorithm}
    \end{minipage}
\end{figure*}

\begin{figure*}
    \centering
    \begin{minipage}{0.32\textwidth}
        \centering
        \includegraphics[width=\linewidth]{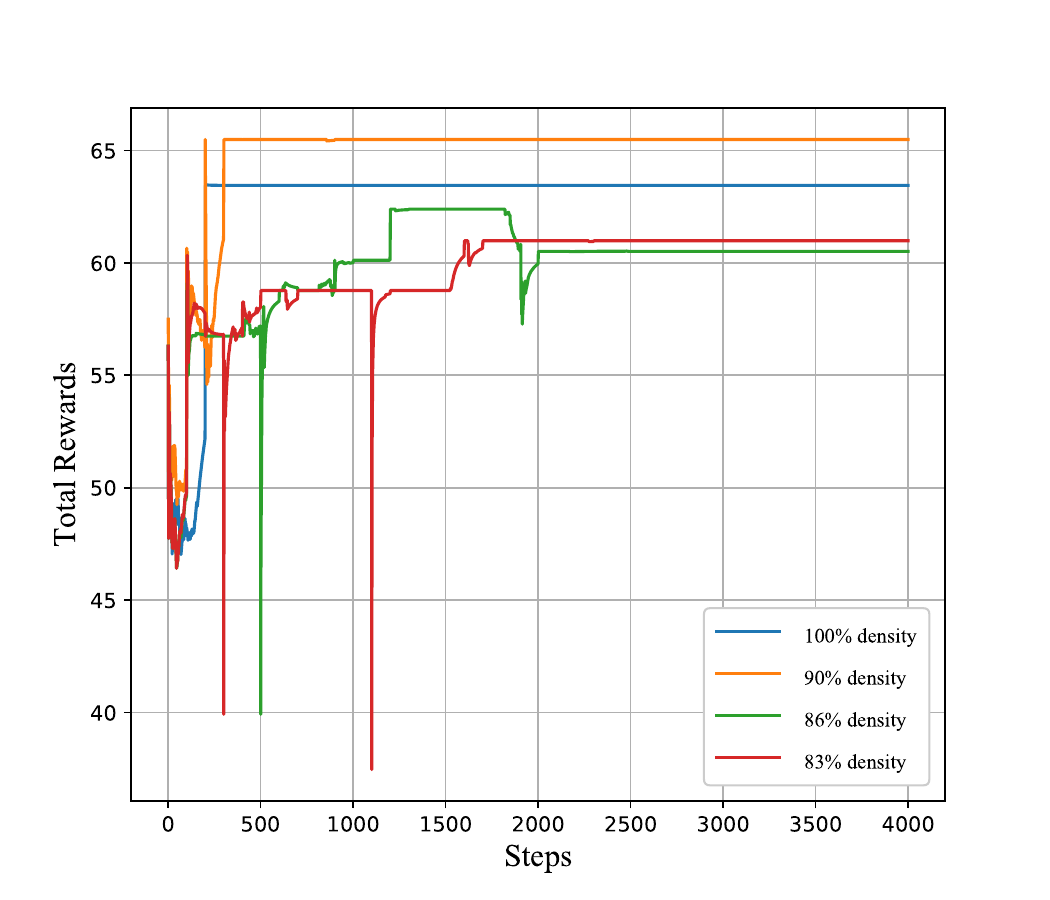}
        \caption{Relation between the total reward and density in the trivial pruning rate threshold.}
        \label{trivial}
    \end{minipage}
    \hfill
    \begin{minipage}{0.32\textwidth}
        \centering
        \includegraphics[width=\linewidth]{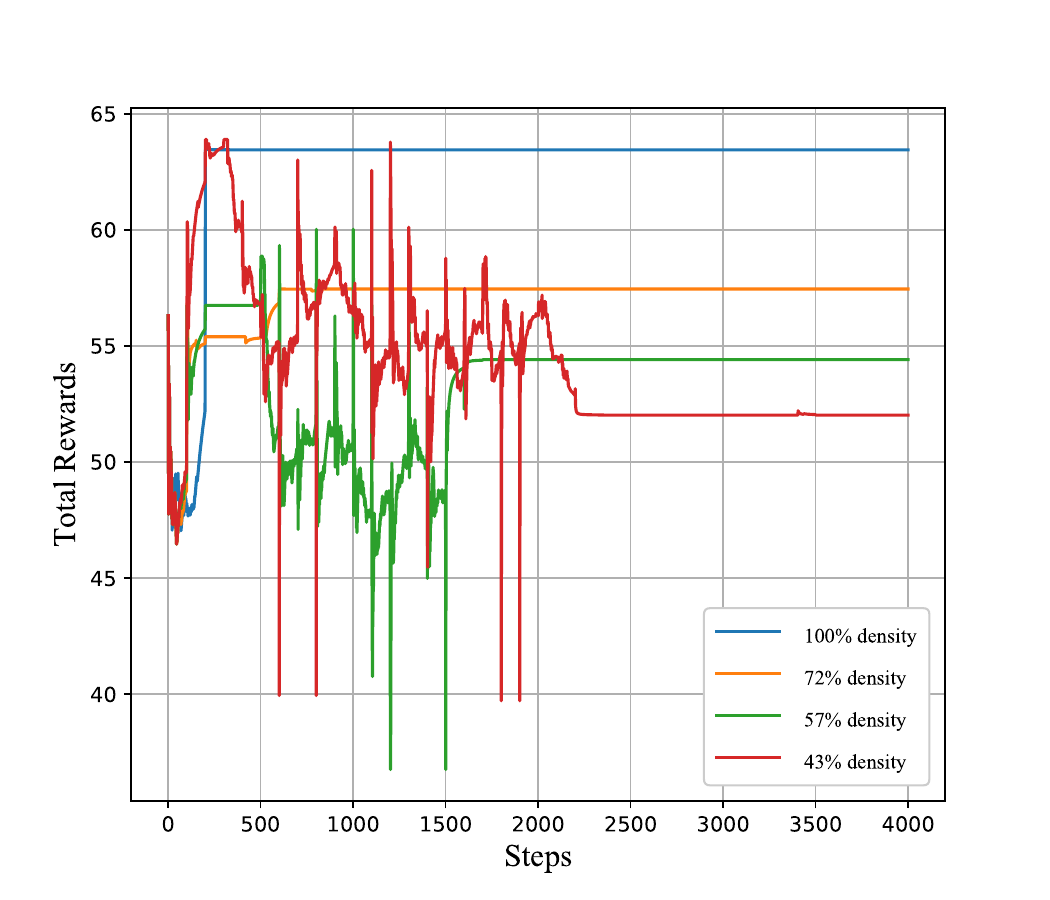}
        \caption{Relation between the total reward and density in the mild pruning rate threshold.}
        \label{mild}
    \end{minipage}
    \hfill
    \begin{minipage}{0.32\textwidth}
        \centering
        \includegraphics[width=\linewidth]{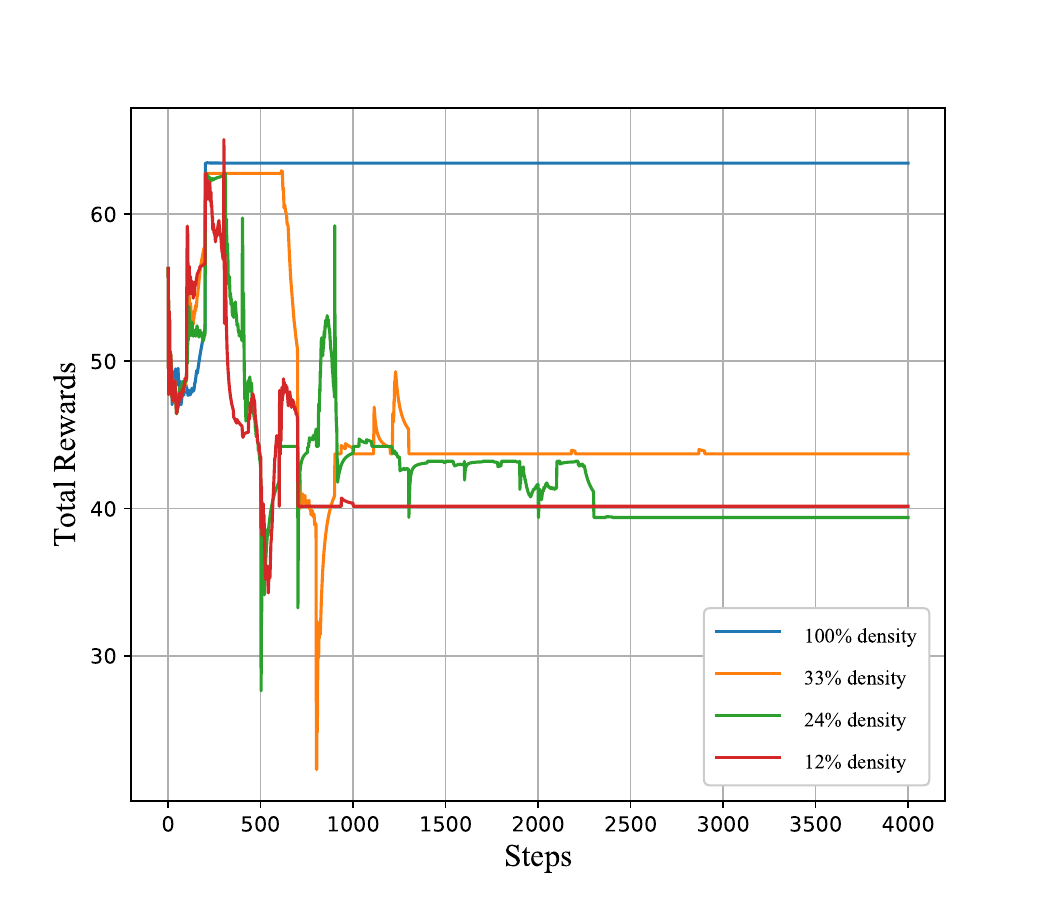}
        \caption{Relation between the total reward and density in the extreme pruning rate threshold.}
        \label{extreme}
    \end{minipage}
\end{figure*}

\section{Numerical Results} \label{numerical_results}
This section demonstrates the effectiveness of the proposed framework by experiments. All experiments are conducted on an NVIDIA Jetson Orin Nano Developer Kit embedded platform running PyTorch 2.3.0 framework within an Ubuntu 22.04 LTS operating environment. For the convergence analysis, we compare our MABLPPO algorithm with three baseline algorithms, which are the MAPPO \cite{yu2022surprising}, the Asynchronous Advantage Actor Critic algorithm for the multi-agent environment (MAA3C) \cite{mnih2016asynchronous} and the random strategies for all players to act. Our experiments divide AVs into three intervals according to the computation resources. We set the pruning rate threshold for each tier of AVs as shown in Table \ref{table1}.


\begin{table}[h]
    \centering 
    \caption{Computation Resource-Density Threshold Correspondence for AVs}
    \label{table1}
    \begin{tabular}{cc} 
        \toprule 
        \textbf{AVs' computation resources} & \textbf{Density threshold} \\
        \midrule 
        High & $(80\%, 100\%]$ \\
        Medium & $(40\%, 80\%]$ \\
        Low & $(0\%, 40\%]$ \\
        \bottomrule 
    \end{tabular}
\end{table}

Figure \ref{before_prune} evaluates the impact of different DRL algorithms on the total rewards. According to the total rewards, MABLPPO achieves the highest total rewards compared to the baseline algorithms, demonstrating that our algorithm is effective in optimizing the bandwidth resource allocation. The Bi-LSTM module can dynamically explore the rules in the time-series data. Furthermore, the curves in Fig. \ref{before_prune} show that MABLPPO can find the optimal strategies faster than others, which spend fewer resources to train the model.

Figure \ref{after_prune} illustrates the impact of the different DRL algorithms pruned by the PX algorithm with 90\% density on the total rewards. After pruning, the performance of TMABLPPO is better than that of the models without pruning. This indicates that the PX algorithm can eliminate less critical neural connections while preserving essential computation pathways.

Figure \ref{compare_prune_algorithm} presents the comparison of pruning algorithms on MABLPPO model, demonstrating that PX attains greater cumulative rewards at 90\% model density than PoPS, Lasso Regression, and Dropout algorithms. The performance plots demonstrate that the models pruned by PX and PoPS algorithms perform better than the baseline model in cumulative rewards. In comparison, other algorithms exhibit lower performance than the baseline model. PX is especially effective at preserving salient neural pathways, allowing for not only more reward harvesting but also quicker fitting than conventional regularization-based pruning approaches.


Figure \ref{trivial} presents the impact of different pruning rates of TMABLPPO models on the total rewards in the trivial pruning rate. The results reveal that the model with 90\% density has higher total rewards than other density settings. The curves in Fig. 6 show that the stability of the model decreased after pruning, demonstrating that the lower density may break the structure of the model.


Comparative analysis in Fig. \ref{mild} examines moderate pruning rates, where the 72\% density model shows 13.9\% lower total rewards than its 90\% density counterpart. This performance degradation suggests moderate pruning may inadvertently remove functionally significant neural connections. The curves in Fig. \ref{mild} show that the observed reduction in learning efficiency implies compromised information processing capacity.

 
Figure \ref{extreme} explores extreme pruning scenarios, where even at 33\% density, the model retains 68\% of its original performance. This demonstrates remarkable network resilience to aggressive parameter reduction, though the significant performance gap from baseline models emphasizes the critical balance required between the accuracy and latency of the models across different pruning intensities.

\begin{figure} [t]
    \centering{
    \includegraphics[width=0.8\linewidth]{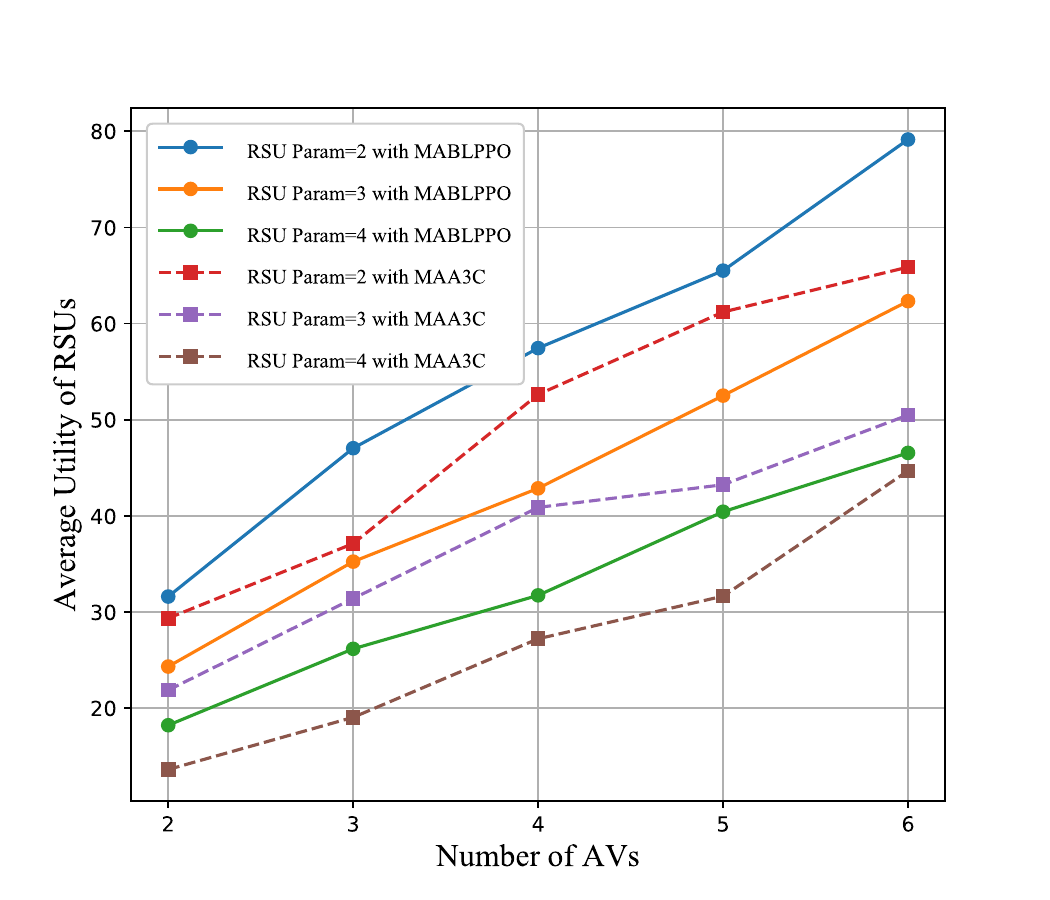}}
    \caption{Average utility of the RSUs with MABLPPO and MAA3C algorithms.}
    \label{fig:rsu_with_a3c_mar26}
\end{figure}

\begin{figure} [t]
    \centering{
    \includegraphics[width=0.8\linewidth]{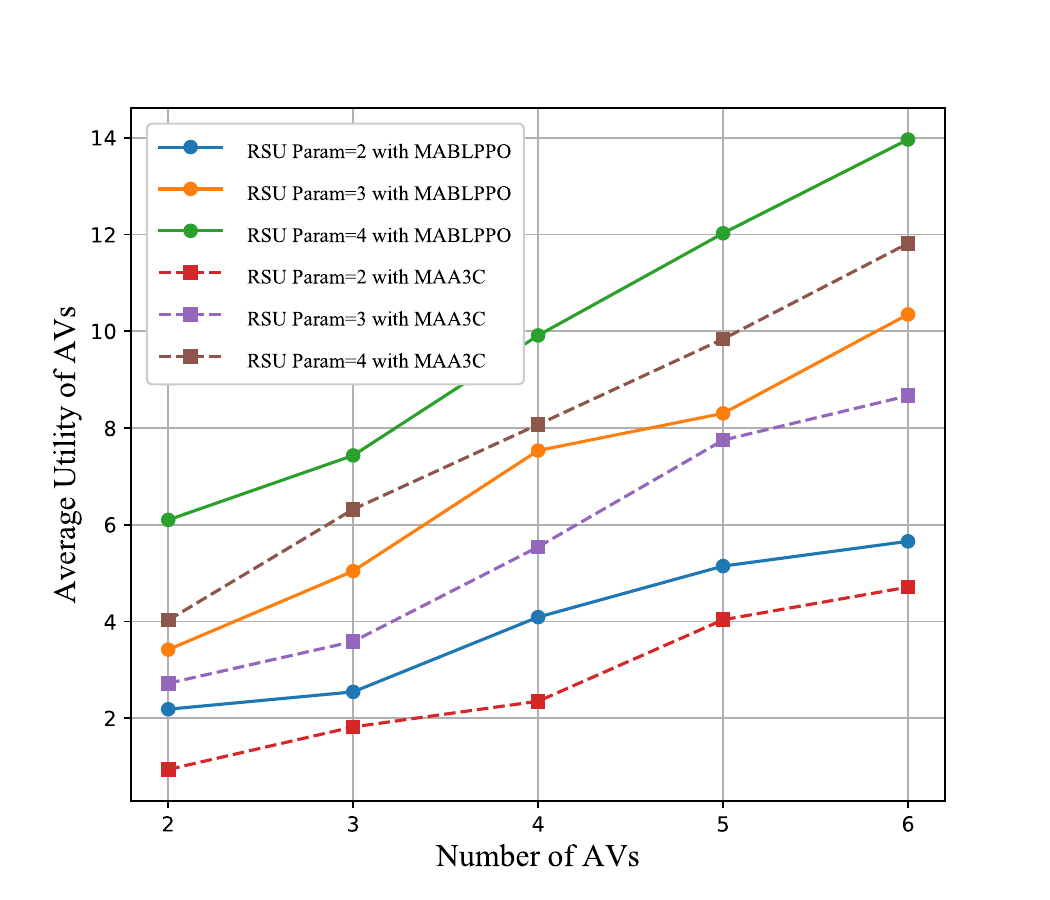}}
    \caption{Average utility of the AVs with MABLPPO and MAA3C algorithms.}
    \label{fig:av_with_a3c_mar26}
\end{figure}

Figure \ref{fig:rsu_with_a3c_mar26} illustrates the influences on the average utility of RSUs. When maintaining a fixed AV count, the average utility of RSUs has an inverse relationship with the number of RSUs. Compared to MAA3C, MABLPPO has better strategies to enhance the utility of RSUs. According to the social effect, the increasing number of RSUs alleviates the lack of bandwidth. Hence, the utilities of RSUs decline. On the other hand, we consider the fixed RSU count, in which the average utility of RSUs grows with the number of AVs.  According to the resource dilution model, the number of AVs increases, compelling intensified bandwidth competition. This competition drives the AVs’ increasing costs to acquire the bandwidth for VEAAT migration, which raises the utilities of RSUs.

Figure \ref{fig:av_with_a3c_mar26} demonstrates the impact of network resource allocation on AVs' average utility, where the MABLPPO algorithm exhibits strategic optimization for AVs. With fixed RSU deployment, the total rewards increase with AV population growth as vehicles actively acquire additional bandwidth to enhance VEAAT migration services, facilitated by sufficient RSU bandwidth capacity. Conversely, under constant AV density, the aggregate utility rises with RSU quantity expansion, directly attributable to the amplified total available bandwidth resources provided by the augmented RSU infrastructure.

\section{Conclusion} \label{conclusion}
In this paper, we studied the service migration issues in vehicular embodied AI networks. We designed a novel algorithm called TMABLPPO for the vehicular embodied AI Twins migration problem in complex traffic environments. We first formulated the resource allocation problem as an MLMF Stackelberg game considering with QoS between AVs and RSUs.  Subsequently, we proposed an enhanced MADRL algorithm based on Bi-LSTM to improve the utilization of historical data. Furthermore, considering the differences in the AVs computation resources, we compressed the actor models by the computation-aware pruning algorithm to balance their latency and performance. Numerical results demonstrate that our proposed approach exhibits notable advantages in terms of performance and latency. In future work, we will focus on optimizing the construction of the embodied AI systems, which integrate the state-of-the-art few-shot learning techniques with large-scale models to solve the diversity problems in intelligent transportation systems.

\bibliographystyle{IEEEtran}

\bibliography{ref}

\end{document}